\documentclass[11pt,reqno]{amsart}
\usepackage[T1]{fontenc}
\usepackage[utf8]{inputenc}
\usepackage{lmodern}
\usepackage[letterpaper,margin=1in]{geometry}
\usepackage{amsmath,amssymb,amsthm}
\usepackage{microtype}
\usepackage[hidelinks,unicode]{hyperref}
\input{glyphtounicode}
\hypersetup{
  pdftitle={Dimension-free exponential entropy quotients over F2},
  pdfauthor={Josiah Park},
  pdfsubject={Entropy quotients and additive combinatorics},
  pdfkeywords={Shannon entropy, finite fields, entropy quotients, additive combinatorics}
}
\allowdisplaybreaks[2]
\providecommand{\tightlist}{\setlength{\itemsep}{0pt}\setlength{\parskip}{0pt}}
\theoremstyle{plain}
\newtheorem{theorem}{Theorem}[section]
\newtheorem{lemma}[theorem]{Lemma}
\newtheorem{proposition}[theorem]{Proposition}
\newtheorem{corollary}[theorem]{Corollary}

\title[Dimension-free exponential entropy quotients]{Dimension-free exponential entropy quotients over \(\mathbb F_2\)}
\author{Josiah Park}
\date{28 October 2026}

\begin{document}
\begin{abstract}
Let \(X\) be a random vector in \(\mathbb F_2^m\), and let \(X'\) be an independent copy. We prove that, for every \(\varepsilon>0\), there is a subspace \(G\) with \(\dim G\le 2^{320(1+\varepsilon)/\varepsilon}H(X)\) such that \(H(q_GX)\le(1+\varepsilon)(H(q_GX+q_GX')-H(q_GX))\). In particular, the quotient entropy is at most \((1+\varepsilon)(H(X+X')-H(X))\). The coefficient is independent of the ambient dimension, and its exponential order in \(1/\varepsilon\) is necessary. The structural result is a logarithmic rigidity estimate for finite weighted families of laws with averaged subspace entropy sparsity. A common variational quotient controls both a family and its convolution family without decreasing the original family's normalized mean entropy defect. Posterior families preserve averaged sparsity, and a fresh-index argument combines a local-to-global estimate with the entropic Marton theorem to increase the defect ratio at a fixed multiplicative sparsity cost. We also obtain two-input almost-additivity and concentration of moderate-doubling sets on subspace cosets at the same exponential scale. The constants are not optimized.
\end{abstract}
\maketitle

\section{Introduction}\label{sec:1}

All random variables in this paper take values in finite binary vector spaces, and all entropies are Shannon entropies in bits. For a law \(X\), with \(X'\) an independent copy, define its entropy increment and entropy defect by
\begin{equation}\label{eq:1}
D(X)=H(X+X')-H(X),\qquad
 E(X)=2H(X)-H(X+X')=H(X)-D(X).
\end{equation}
For a subspace \(G\le\mathbb F_2^m\), write \(q_G:\mathbb F_2^m\to\mathbb F_2^m/G\) for the quotient map. We prove that an entropy-sized subspace, with a dimension-independent exponential coefficient, makes the remaining entropy nearly additive under independent doubling.

\begin{theorem}[Exponential entropy quotient]\label{thm:1.1}

For every \(\varepsilon>0\) and every law \(X\) on \(\mathbb F_2^m\), there is a subspace \(G\) such that
\begin{equation}\label{eq:2}
\dim G\le 2^{320(1+\varepsilon)/\varepsilon}H(X)
\end{equation}
and
\begin{equation}\label{eq:3}
H(q_GX)\le(1+\varepsilon)D(q_GX)
                    \le(1+\varepsilon)D(X).
\end{equation}
The constant is absolute and is independent of \(m\).

\end{theorem}

The first inequality in \eqref{eq:3} is stronger than a bound using only the original increment \(D(X)\). No independence among the coordinates of \(X\) is assumed. In particular, the theorem applies to every noninjective linear image of an arbitrarily biased product law, without a restriction on the number or rank density of its representing columns.

The dimension-uniform exponential scale is the point of the theorem. The number 320 is a convenient absolute constant, not an optimal one. Proposition~\ref{thm:9.3} shows that any universally valid coefficient \(C(\varepsilon)\) in the original-input version of \eqref{eq:2}--\eqref{eq:3} must satisfy
\begin{equation}\label{eq:4}
\liminf_{\varepsilon\downarrow0}
 \varepsilon\log_2 C(\varepsilon)\ge2.
\end{equation}
Thus exponential dependence is necessary, but the optimal unrestricted leading constant is not determined here.

\subsection{Context and normalization}\label{sec:1.1}

Abbe, Sandon, Shashkov, and Viazovska \cite[Conjecture 7.1]{ref3} formulate a related entropy assertion using smoothing by a uniform subspace variable. Their displayed coefficient is written \(\exp(O_m(1/\varepsilon))\). We distinguish that notation from the absolute, dimension-independent coefficient proved in Theorem~\ref{thm:1.1}. If \(U_G\) is independent and uniform on \(G\), then
\begin{equation}\label{eq:5}
H(X+U_G)-H(U_G)=H(q_GX),
\end{equation}
because, conditional on a quotient value, adding \(U_G\) makes the law uniform on the corresponding coset. Consequently the quotient and uniform-smoothing formulations are equivalent with these normalizations.

Cohen and Zakharov \cite[Theorems 1.2 and 5.5]{ref2} prove a dimension-independent almost-additivity theorem for two input laws. Their quantitative choice has the form \[
 L(\alpha)=\exp\bigl(O(\alpha^{-1}\log(1/\alpha))\bigr).
\] More specifically, \cite[Corollary 1.3, equation (1.2)]{ref2} reads in our notation \[
 H(q_GU)\le H(U+V)-H(V)+\alpha\bigl(H(U)+H(V)\bigr),
\] with \(\dim G\le L(\alpha)(H(U)+H(V))\), for independent inputs \(U,V\). By contrast, Corollary~\ref{thm:9.2} gives \[
 H(q_GU)\le H(U+V)-H(V)
             +\alpha\bigl(H(q_GU)+H(q_GV)\bigr)
\] with \(\dim G\le\tfrac12 2^{320/\alpha}(H(U)+H(V))\), and the same conclusion with the inputs interchanged. The error is measured by the remaining quotient entropies, which do not exceed the original entropies. The main quantitative improvement is the removal of the factor \(\log(1/\alpha)\) in the exponent.

For two independent copies of one law, the earlier display gives \(H(q_GX)\le D(X)+2\alpha H(X)\), whereas Theorem~\ref{thm:1.1} gives the relative and intrinsic bound \eqref{eq:3}. The earlier almost-additivity theorem can also be iterated to obtain a relative-error bound, but with the weaker quantitative scale; we include that entropy-halving comparison in Section~\ref{sec:9.1}. Thus the point is the improved exponential order, not the first dimension-independent existence theorem. All comparisons of exponential orders concern the small-error limit, not a claim that the explicit constant 320 is better at every fixed error tolerance.

The only external structural theorem used in the proof of Theorem~\ref{thm:1.1} is the entropic polynomial Freiman--Ruzsa theorem of Gowers, Green, Manners, and Tao \cite[Theorem 1.2]{ref1}. We use its individual constant-six consequence. The required conditional Balog--Szemerédi--Gowers entropy inequality, also appearing in \cite[Lemma A.2]{ref1}, is proved in Section~\ref{sec:6.2}. All numbered source locators refer to the versions specified in the bibliography.

\subsection{Sparse laws and the proof mechanism}\label{sec:1.2}

A law \(Z\) on a binary vector space is called \(\lambda\)-sparse in entropy if
\begin{equation}\label{eq:6}
H(Z\mid q_WZ)\le\lambda\dim W
 \quad\text{for every subspace }W.
\end{equation}
Minimizing \(H(q_GX)+\lambda\dim G\) produces such a quotient at dimension cost at most \(H(X)/\lambda\). The central estimate is the following.

\begin{theorem}[Sparse-law rigidity]\label{thm:1.2}

Let \(0<\lambda<1\). Every \(\lambda\)-sparse law \(Z\) with \(H(Z)>0\) satisfies
\begin{equation}\label{eq:7}
\frac{E(Z)}{H(Z)}
 \le\min\left\{1,\frac{320}{\log_2(1/\lambda)}\right\}.
\end{equation}

\end{theorem}

We prove a stronger statement for finite weighted families of laws. The sparsity condition is averaged over the family index, rather than required separately of every component. This distinction is essential: posterior conditioning preserves averaged sparsity even when individual posterior laws have large within-subspace entropy.

Two constructions drive the proof. First, a common variational quotient controls the sparsity of both a family and its convolution family without decreasing the normalized mean entropy defect. Second, a local-to-global estimate samples a fresh family index independently of a subspace built from previous samples. This permits averaged sparsity to be used against subspaces that depend on the earlier indices. Four-sample identities and the entropic Marton theorem then increase the defect ratio by a fixed multiple of its square while multiplying the sparsity parameter by at most four. Since all defect ratios are at most one, iteration gives \eqref{eq:7}.

The entropy framework builds on \cite{ref1} and \cite{ref2}. In particular, the nested interaction inequality is the one in \cite[Lemma 5.3]{ref2}, and the random-subspace accumulation argument follows the local-to-global method of \cite[Lemma 5.2]{ref2}. We give the needed calculations in full. The distinction in the present argument is to work throughout with posterior-closed averaged sparsity and to combine these calculations with a ratio-preserving common quotient, so that each ratio-growth step has a fixed multiplicative sparsity cost.

The argument is an existence proof. The family index sets may grow rapidly, and the variational steps minimize over all subspaces. No efficient algorithm for finding the quotient is asserted.

\subsection{Significance and relation to the predecessor manuscript}\label{sec:1.3}

The conclusion separates the dimension of the ambient space from the information content of the law. For each fixed error tolerance, the number of directions that must be removed is bounded by a constant times \(H(X)\), regardless of the number of coordinates. In the remaining quotient, independent addition nearly doubles entropy: if \(Z=q_GX\), then \eqref{eq:3} is equivalent to \[
 H(Z+Z')\ge\left(2-\frac{\varepsilon}{1+\varepsilon}\right)H(Z).
\] Thus the additive entropy defect can be confined to a subspace at an entropy-proportional dimension cost. This interpretation does not assert that the quotient coordinates are independent, nor that the quotient is close to a uniform law.

Quantitatively, the bound removes the factor \(\log(1/\varepsilon)\) in the exponent of the cited general comparison estimate. Proposition~\ref{thm:9.3} shows that the resulting scale \(2^{\Theta(1/\varepsilon)}\) is the correct one for the least universal coefficient as \(\varepsilon\downarrow0\). Only the scale is settled: the lower leading constant is at least two, whereas the present upper constant is 320. For moderate-doubling sets, Corollary~\ref{thm:9.4} translates the entropy theorem into a subspace-coset concentration statement with the corresponding exponential dependence on the error parameter. These are structural existence results, not efficient algorithms or practical small-constant estimates.

The author's predecessor manuscript \cite[Section 12.1]{ref4} develops the entropy-sparsity formulation and single-law variational reduction used here, and also studies products and structured linear images. The present article incorporates the necessary definition and proves the variational reduction in full, in the stronger weighted-family form of Lemma~\ref{thm:9.1}. The sparse-law rigidity estimate supplies the missing premise of that reduction for unrestricted laws. No result from the unpublished manuscript is assumed: its role here is attribution and context. The specialized quantitative bounds for products and structured images remain separate from the nonsharp universal bound proved here. The only external structural theorem on which the proof of Theorem~\ref{thm:1.1} depends is the entropic Marton theorem, stated in the form used in Section~\ref{sec:7}.

\section{Entropy tools and weighted families}\label{sec:2}

\subsection{Laws, convolution, and subspace entropy}\label{sec:2.1}

For a finite-valued variable \(A\) with mass function \(p\), \[
 H(A)=-\sum_a p(a)\log_2 p(a),\qquad 0\log_2 0=0.
\] Conditional entropy is the average of the entropies of the conditional laws. We use the chain rule and the nonnegativity of conditional mutual information, \[
 I(A;B\mid C)=H(A\mid C)+H(B\mid C)-H(A,B\mid C)\ge0.
\] We write \(H(P)\) for the entropy of a law \(P\). Entropy is invariant under translations and bijective changes of coordinates.

For laws \(P,Q\) on a common binary vector space \(V\), let \(P*Q\) denote the law of \(A+B\) for independent \(A\sim P\) and \(B\sim Q\). Define
\begin{equation}\label{eq:8}
s(P,Q)=H(P)+H(Q)-H(P*Q),
 \qquad
 \mathsf d(P,Q)=H(P*Q)-\frac{H(P)+H(Q)}2.
\end{equation}
The latter is entropic Ruzsa distance; subtraction and addition coincide in characteristic two. Since adding independent noise does not decrease entropy,
\begin{equation}\label{eq:9}
0\le s(P,Q)\le\min\{H(P),H(Q)\}.
\end{equation}
In particular, \(s(P,P)=E(P)\), \(\mathsf d(P,P)=D(P)\), and \(0\le D(P),E(P)\le H(P)\).

For a subspace \(W\le V\), put
\begin{equation}\label{eq:10}
F_P(W)=H(P)-H(q_WP).
\end{equation}
If \(A\sim P\), then \(F_P(W)=H(A\mid q_WA)\). The quantity is nonnegative and nondecreasing under inclusion of subspaces.

\begin{lemma}[Subspace entropy inequality]\label{thm:2.1}

For subspaces \(U,W\le V\),
\begin{equation}\label{eq:11}
H(q_UP)+H(q_WP)
 \ge H(q_{U\cap W}P)+H(q_{U+W}P).
\end{equation}
Equivalently, \(F_P\) is supermodular.

\end{lemma}

\begin{proof}

The pair \((q_UA,q_WA)\) determines \(q_{U\cap W}A\), and \(q_{U+W}A\) is a common function of the two components. If \(C\) is a common function of variables \(B_1,B_2\), then \[
 H(B_1)+H(B_2)-H(B_1,B_2)-H(C)
 =I(B_1;B_2\mid C)\ge0.
\] Applying this identity gives \eqref{eq:11}.
\end{proof}

\subsection{Weighted families}\label{sec:2.2}

A finite weighted family \(\mathcal P=(P_i,a_i)_{i\in I}\) consists of laws on one common space \(V\), with positive weights summing to one. Zero-weight components may always be omitted. Define
\begin{equation}\label{eq:12}
\begin{aligned}
 h(\mathcal P)&=\sum_i a_iH(P_i),\\
 e(\mathcal P)&=\sum_{i,j}a_ia_j s(P_i,P_j),\\
 F_{\mathcal P}(W)&=\sum_i a_iF_{P_i}(W)
                 =h(\mathcal P)-h(q_W\mathcal P).
 \end{aligned}
\end{equation}
Here \(q_W\mathcal P\) denotes the componentwise quotient with the same weights. We say that \(\mathcal P\) is \(\lambda\)-sparse if
\begin{equation}\label{eq:13}
F_{\mathcal P}(W)\le\lambda\dim W
 \quad\text{for every }W\le V.
\end{equation}
This is an averaged condition, not a componentwise one. Also, \(h(\mathcal P)\) is the mean component entropy, not the entropy of the mixture \(\sum_i a_iP_i\). Equation \eqref{eq:9} gives
\begin{equation}\label{eq:14}
0\le e(\mathcal P)\le h(\mathcal P).
\end{equation}

The convolution family \(\mathcal R=\mathcal P*\mathcal P\) has components \(P_i*P_j\) and weights \(a_ia_j\). Directly from \eqref{eq:8},
\begin{equation}\label{eq:15}
h(\mathcal R)=2h(\mathcal P)-e(\mathcal P).
\end{equation}

\subsection{Posterior families}\label{sec:2.3}

Draw independent indices \(I,J\) with weights \(a_i\), and draw \(A\sim P_I\), \(B\sim P_J\) independently conditional on the indices. The posterior family \(\operatorname{Post}(\mathcal P)\) has components
\begin{equation}\label{eq:16}
Q_{i,j,t}=\operatorname{Law}(A\mid I=i,J=j,A+B=t)
\end{equation}
with weights \(a_ia_j(P_i*P_j)(t)\), omitting zero-weight indices.

\begin{lemma}[Posterior closure]\label{thm:2.2}

If \(\mathcal Q=\operatorname{Post}(\mathcal P)\), then
\begin{equation}\label{eq:17}
h(\mathcal Q)=e(\mathcal P),\qquad
 F_{\mathcal Q}(W)\le F_{\mathcal P}(W)
 \quad(W\le V).
\end{equation}
In particular, posterior formation preserves averaged \(\lambda\)-sparsity.

\end{lemma}

\begin{proof}

Given the two indices, either input is determined by the other input and their sum. Thus \[
 h(\mathcal Q)=H(A\mid I,J,A+B)=e(\mathcal P).
\] For every fixed \(W\), conditioning gives \[
 \begin{aligned}
 F_{\mathcal Q}(W)
 &=H(A\mid I,J,A+B,q_WA)\\
 &\le H(A\mid I,q_WA)=F_{\mathcal P}(W).
 \end{aligned}
\] The averaging over all posterior indices is part of the identity; no sparsity assertion about an individual posterior law is used.
\end{proof}

\section{Internal interaction in nested subspaces}\label{sec:3}

We use the fibring framework of \cite[Section 2.1]{ref2}. The nested inequality \eqref{eq:22} is \cite[Lemma 5.3]{ref2}; the proof below records it in our notation.

For independent \(A\sim P\), \(B\sim Q\), define
\begin{equation}\label{eq:18}
\begin{aligned}
 s_W(P,Q)={}&H(A\mid q_WA)+H(B\mid q_WB)\\
            &-H(A+B\mid q_WA,q_WB).
 \end{aligned}
\end{equation}
This is the average of \(s\) between independent conditional fiber laws. Consequently,
\begin{equation}\label{eq:19}
0\le s_W(P,Q)\le\min\{F_P(W),F_Q(W)\}.
\end{equation}

\begin{lemma}[Fibring and nested interaction]\label{thm:3.1}

For every subspace \(W\),
\begin{equation}\label{eq:20}
\begin{aligned}
 s(P,Q)={}&s(q_WP,q_WQ)+s_W(P,Q)\\
 &-I(A+B;q_WA,q_WB\mid q_W(A+B)).
 \end{aligned}
\end{equation}
Hence
\begin{equation}\label{eq:21}
s_W(P,Q)\ge s(P,Q)-s(q_WP,q_WQ).
\end{equation}
If \(U\subseteq W\), then
\begin{equation}\label{eq:22}
F_Q(U)\ge s_W(P,Q)-\bigl[H(q_UP)-H(q_WP)\bigr].
\end{equation}

\end{lemma}

\begin{proof}

Expand the three terms defining \(s(P,Q)\) by their quotient variables. The difference between conditioning \(A+B\) on its quotient and conditioning it on both input quotients is the mutual information in \eqref{eq:20}. This proves \eqref{eq:20} and \eqref{eq:21}.

For completeness, we give the nested calculation explicitly. Let \(R=q_UA\), \(S=q_UB\), \(T=q_WA\), and \(L=q_WB\). The variables \(T,L\) are functions of \(R,S\), respectively. Expanding the first two terms in \eqref{eq:18} by \(R,S\), and conditioning the last residual entropy further on \((R,S)\), gives
\begin{equation}\label{eq:23}
\begin{aligned}
 s_W(P,Q)\le{}&s_U(P,Q)\\
 &+H(R\mid T)+H(S\mid L)-H(R+S\mid T,L).
 \end{aligned}
\end{equation}
For fixed \((T,L)\), the laws of \(R,S\) are independent. The second line of \eqref{eq:23} is therefore an average interaction, bounded by \(H(R\mid T)\) using \eqref{eq:9}. Also, \(s_U(P,Q)\le F_Q(U)\) by \eqref{eq:19}. Since \[
 H(R\mid T)=H(q_UP)-H(q_WP),
\] equation \eqref{eq:22} follows.
\end{proof}

\section{A common variational quotient}\label{sec:4}

\begin{lemma}[Ratio-preserving sparsification]\label{thm:4.1}

Suppose \(\mathcal P\) is \(\lambda\)-sparse, \(\lambda>0\), with \(h=h(\mathcal P)>0\). Set \(d=e(\mathcal P)/h\). There is a quotient family \(\mathcal P_0\) with positive mean entropy such that

\begin{enumerate}
\def\labelenumi{\arabic{enumi}.}
\tightlist
\item
  \(e(\mathcal P_0)/h(\mathcal P_0)\ge d\);
\item
  \(\mathcal P_0\) is \(2\lambda\)-sparse;
\item
  \(\mathcal P_0*\mathcal P_0\) is \(4\lambda\)-sparse.
\end{enumerate}

\end{lemma}

\begin{proof}

Put \(\mathcal R=\mathcal P*\mathcal P\), and minimize over subspaces \(G\le V\) the finite objective
\begin{equation}\label{eq:24}
\Phi(G)=h(q_G\mathcal R)+(2+d)h(q_G\mathcal P)
                       +4\lambda\dim G.
\end{equation}
Among minimizers choose one of least dimension. Write \[
 f=F_{\mathcal P}(G),\qquad g=F_{\mathcal R}(G).
\] Equation \eqref{eq:15} gives \(\Phi(0)=4h\). Comparison with zero yields \[
 g+(2+d)f\ge4\lambda\dim G.
\] The input sparsity gives \(f\le\lambda\dim G\). Since
\begin{equation}\label{eq:25}
\begin{aligned}
 e(q_G\mathcal P)-d\,h(q_G\mathcal P)
 &=g-(2-d)f\\
 &\ge4(\lambda\dim G-f)\ge0,
 \end{aligned}
\end{equation}
the ratio does not decrease whenever the output mean entropy is positive.

It cannot vanish. If \(h(q_G\mathcal P)=0\), then every quotient component is deterministic, so \(h(q_G\mathcal R)=0\) as well. Moreover, \(h=f\le\lambda\dim G\). Hence \[
 \Phi(G)=4\lambda\dim G\ge4h=\Phi(0).
\] Minimality would make zero a minimizer too, contradicting the least-dimension choice of a zero-entropy quotient when \(h>0\).

For any subspace \(W\le V/G\), compare \(G\) with the inverse image of \(W\), whose dimension is \(\dim G+\dim W\). Minimality gives
\begin{equation}\label{eq:26}
F_{q_G\mathcal R}(W)+(2+d)F_{q_G\mathcal P}(W)
 \le4\lambda\dim W.
\end{equation}
Both terms are nonnegative. Thus \(q_G\mathcal P\) is \(4\lambda/(2+d)\)-sparse, and \(q_G\mathcal R\) is \(4\lambda\)-sparse. Since \(d\ge0\) and quotienting commutes with convolution, these are the stated bounds.
\end{proof}

The common objective in \eqref{eq:24} is important. An arbitrary quotient of a sparse family need not satisfy the required new sparsity bound. The lemma simultaneously controls the family and the convolution family while retaining the normalized defect.

\section{An averaged local-to-global estimate}\label{sec:5}

The following estimate is an averaged-sparsity form of the local-to-global sampling argument in \cite[Lemma 5.2]{ref2}. Its right-hand side involves mean component entropy, rather than the entropy of the mixture.

\begin{lemma}[Averaged density estimate]\label{thm:5.1}

Let \(\rho\ge0\), and let \(\mathcal Q=(Q_u,b_u)\) be a \(\rho\)-sparse family with mean entropy \(h\). For each ordered pair \((u,w)\), choose a subspace \(V(u,w)\). With independent family indices, put
\begin{equation}\label{eq:27}
A=\mathbb E_{u,w}s_{V(u,w)}(Q_u,Q_w),\qquad
 B=\mathbb E_{u,w}\dim V(u,w).
\end{equation}
Then for every integer \(N\ge1\),
\begin{equation}\label{eq:28}
A\le\frac hN+\frac{\rho B(N-1)}2.
\end{equation}
If \(B\le26h\), it follows that
\begin{equation}\label{eq:29}
A\le14\sqrt\rho\,h.
\end{equation}

\end{lemma}

\begin{proof}

Draw \(u,w_1,w_2,\ldots\) independently from the family-index law. Hold \(u\) fixed within this sequence, and set \[
 S_k=\sum_{j=1}^k V(u,w_j),\qquad S_0=0,
 \qquad h_k=\mathbb E H(q_{S_k}Q_u).
\] Then \(h_0=h\) and \(h_k\ge0\). Given \(u,w_1,\ldots,w_k\), draw the fresh index \(w=w_{k+1}\). Apply \eqref{eq:22} with \[
 P=Q_u,\quad Q=Q_w,\quad
 W=V(u,w),\quad U=S_k\cap V(u,w).
\] Monotonicity of \(F_{Q_w}\), followed by the subspace entropy inequality \eqref{eq:11}, gives
\begin{equation}\label{eq:30}
\begin{aligned}
 F_{Q_w}(S_k)\ge{}&s_{V(u,w)}(Q_u,Q_w)\\
 &-\bigl[H(q_{S_k}Q_u)-H(q_{S_k+V(u,w)}Q_u)\bigr].
 \end{aligned}
\end{equation}
To check the direction, \eqref{eq:11} gives \[
 H(q_{S_k\cap V}Q_u)-H(q_VQ_u)
 \le H(q_{S_k}Q_u)-H(q_{S_k+V}Q_u),
\] which is precisely the replacement used in \eqref{eq:30}.

The subspace \(S_k\) does not depend on the fresh \(w\). Conditional on all previously drawn indices, averaging the left side of \eqref{eq:30} over \(w\) therefore gives \[
 F_{\mathcal Q}(S_k)\le\rho\dim S_k.
\] Averaging also over the earlier indices yields \[
 A\le h_k-h_{k+1}+\rho\mathbb E\dim S_k
   \le h_k-h_{k+1}+\rho kB.
\] Summing for \(k=0,\ldots,N-1\), using \(h_N\ge0\), and dividing by \(N\) proves \eqref{eq:28}.

For \(0<\rho\le1\), take \(N=\lceil\rho^{-1/2}\rceil\). Then \(1/N\le\sqrt\rho\) and \(N-1<\rho^{-1/2}\), so \(B\le26h\) gives \eqref{eq:29}. For \(\rho>1\), use \(A\le h\), which follows from \eqref{eq:19}. If \(\rho=0\), sparsity of the full ambient space forces \(h=0\).
\end{proof}

The subspaces in this lemma may depend on both sampled indices. The proof requires only that the new second index be independent of the subspace built from the previous samples. No componentwise sparsity assumption is introduced.

\section{Four-sample identities and conditional additive energy}\label{sec:6}

This section records the weighted-family form of the four-sample entropy framework used in \cite{ref1} and \cite[Section 3]{ref2}. All family-index conditioning is retained explicitly.

Let \(\mathcal P\) have mean entropy \(H>0\) and ratio \(\delta=e(\mathcal P)/H\). Define \[
 \mathcal R=\mathcal P*\mathcal P,\qquad
 \mathcal Q=\operatorname{Post}(\mathcal P).
\] By \eqref{eq:15} and \eqref{eq:17},
\begin{equation}\label{eq:31}
h(\mathcal R)=(2-\delta)H,\qquad
 h(\mathcal Q)=\delta H.
\end{equation}
Draw four independent indexed pairs \((I_r,X_r)\), each with the family-index distribution and conditional law \(X_r\mid I_r=i\sim P_i\). Write \[
 \begin{gathered}
 T=(I_1,I_2,I_3,I_4),\qquad S=X_1+X_2+X_3+X_4,\\
 A=X_1+X_2,\quad B=X_1+X_3,\quad
 C=X_2+X_3=A+B,
 \end{gathered}
\] and put \(J=I(A;B\mid S,T)\ge0\).

\subsection{Splitting the mean defect}\label{sec:6.1}

\begin{lemma}[Four-sample splitting]\label{thm:6.1}

With the preceding definitions,
\begin{equation}\label{eq:32}
e(\mathcal R)+e(\mathcal Q)=2\delta H+J.
\end{equation}

\end{lemma}

\begin{proof}

Conditional on \(T\), the sums \(A\) and \(S+A\) are independent. Averaging their conditional interaction gives \[
 H(A\mid S,T)=e(\mathcal R).
\] The analogous conditional entropies of \(B\) and \(C\) agree after averaging, by exchangeability of the four indexed pairs. Conditional on \(A,S,T\), the first coordinates \(X_1,X_3\) have independent posterior laws. Their posterior indices, before this conditioning, are independent with the distribution defining \(\mathcal Q\). Hence \[
 \begin{aligned}
 e(\mathcal Q)
 &=2h(\mathcal Q)-H(B\mid A,S,T)\\
 &=2\delta H-H(B\mid S,T)+J\\
 &=2\delta H-e(\mathcal R)+J.
 \end{aligned}
\] This is \eqref{eq:32}. The symmetry is an averaged symmetry of the indexed pairs; it is not an assertion that all four component laws agree at a fixed tuple of indices.
\end{proof}

\subsection{A conditional Balog--Szemerédi--Gowers inequality}\label{sec:6.2}

\begin{lemma}\label{thm:6.2}

For possibly dependent binary-vector variables \(U,V\), let \(Z=U+V\). Then
\begin{equation}\label{eq:33}
\begin{aligned}
 &\mathbb E_z\,\mathsf d\bigl(\operatorname{Law}(U\mid Z=z),
                              \operatorname{Law}(V\mid Z=z)\bigr)\\
 &\hspace{16mm}\le3I(U;V)+2H(Z)-H(U)-H(V).
 \end{aligned}
\end{equation}

\end{lemma}

\begin{proof}

Draw two conditionally independent trials \((U_1,V_1),(U_2,V_2)\) given \(Z\). Thus \(U_1+V_1=U_2+V_2=Z\). In characteristic two, \[
 L:=U_1+V_2=U_2+V_1.
\] Entropy submodularity gives \[
 H(L)+H(L,U_1,V_1)\le H(L,U_1)+H(L,V_1).
\] The triple \((L,U_1,V_1)\) determines both trials, and conversely, so \[
 H(L,U_1,V_1)=2H(U,V)-H(Z).
\] Also, \(H(L,U_1)=H(U_1,V_2)\le H(U)+H(V)\), and the analogous bound holds for \(H(L,V_1)\). Therefore \[
 H(L\mid Z)\le H(L)\le2I(U;V)+H(Z).
\] Both \(H(U_1\mid Z)\) and \(H(V_2\mid Z)\) equal \(H(U,V)-H(Z)\). Subtracting their mean gives \eqref{eq:33}, by definition \eqref{eq:8}.
\end{proof}

\begin{lemma}[Posterior doubling estimate]\label{thm:6.3}

For independent posterior-family indices \(u,w\),
\begin{equation}\label{eq:34}
\mathbb E_{u,w}D(Q_u*Q_w)\le3J.
\end{equation}

\end{lemma}

\begin{proof}

Apply Lemma~\ref{thm:6.2} to \((A,C)\), conditional on each value of \((S,T)\), and average. Since \(A+C=B\), the right side becomes \[
 \begin{aligned}
 3I(A;C\mid S,T)&+2H(B\mid S,T)\\
               &-H(A\mid S,T)-H(C\mid S,T).
 \end{aligned}
\] The three conditional entropies agree after averaging, and \(I(A;C\mid S,T)=J\): the pair \((A,C)\) is in bijection with \((A,B)\), while the averaged marginal entropies of \(B,C\) agree. Thus the expression is \(3J\).

Given \((B=b,S=s,T)\), the law of \(C=A+B\) is a translate of the law of \(A\). The Ruzsa distance on the left of \eqref{eq:33} is therefore \(D\) of the conditional law of \(A\). Conditional on these same data, the pairs \((X_1,X_3)\) and \((X_2,X_4)\) remain independent subject to their respective sums \(b\) and \(s+b\). Hence \[
 \operatorname{Law}(A\mid B=b,S=s,T)
 =Q_{I_1,I_3,b}*Q_{I_2,I_4,s+b}.
\] Unconditionally, the two triples \((I_1,I_3,X_1+X_3)\) and \((I_2,I_4,X_2+X_4)\) are independent, and each has the posterior-index distribution. Their average is exactly the left side of \eqref{eq:34}.
\end{proof}

\section{The entropic Marton input}\label{sec:7}

We use the following consequence of \cite[Theorem 1.2]{ref1}. If \(K\) is any law on a binary vector space, there is a subspace \(V\) such that
\begin{equation}\label{eq:35}
\mathsf d(K,U_V)\le6D(K),
\end{equation}
where \(U_V\) is uniform on \(V\). The theorem in \cite{ref1} is stated with natural logarithms; dividing all entropy terms by \(\ln2\) preserves the constant.

Uniform smoothing and the elementary entropy bound for an independent sum give \[
 \mathsf d(K,U_V)=H(q_VK)+\frac{\dim V-H(K)}2,
 \qquad
 \mathsf d(K,U_V)\ge\frac{|\dim V-H(K)|}2.
\] Consequently \eqref{eq:35} implies
\begin{equation}\label{eq:36}
H(q_VK)\le12D(K),\qquad
 \dim V\le H(K)+12D(K)\le13H(K).
\end{equation}

\begin{lemma}[Sparse posterior estimate]\label{thm:7.1}

In the notation of Section~\ref{sec:6}, suppose the posterior family \(\mathcal Q\) is \(\rho\)-sparse. Then
\begin{equation}\label{eq:37}
e(\mathcal Q)\le36J+14\sqrt\rho\,\delta H.
\end{equation}

\end{lemma}

\begin{proof}

For every ordered pair \((u,w)\), apply \eqref{eq:36} to \(K_{u,w}=Q_u*Q_w\), and choose a corresponding subspace \(V(u,w)\). Their mean dimension satisfies
\begin{equation}\label{eq:38}
\mathbb E\dim V(u,w)
 \le13\mathbb E H(Q_u*Q_w)
 \le26h(\mathcal Q).
\end{equation}
By \eqref{eq:21}, \[
 s(Q_u,Q_w)\le s_{V(u,w)}(Q_u,Q_w)
                     +s(q_VQ_u,q_VQ_w).
\] The last term is at most \(H(q_V(Q_u*Q_w))\), since adding independent variables does not decrease entropy and \eqref{eq:9} applies. Equation \eqref{eq:36} bounds it by \(12D(K_{u,w})\). Averaging and applying Lemma~\ref{thm:6.3} gives \[
 e(\mathcal Q)\le
 \mathbb E s_{V(u,w)}(Q_u,Q_w)+36J.
\] Lemma~\ref{thm:5.1} and \eqref{eq:38} bound the first term by \(14\sqrt\rho\,h(\mathcal Q)\). Now use \(h(\mathcal Q)=\delta H\) from \eqref{eq:31}.
\end{proof}

\section{Fixed-cost ratio growth and logarithmic rigidity}\label{sec:8}

\begin{lemma}[Ratio growth]\label{thm:8.1}

Let a \(\lambda\)-sparse family have positive mean entropy and ratio \(d=e/h>0\). If \(d\ge40\sqrt\lambda\), there is a \(4\lambda\)-sparse family with positive mean entropy and ratio at least
\begin{equation}\label{eq:39}
d+\frac{d^2}{148}.
\end{equation}

\end{lemma}

\begin{proof}

Apply Lemma~\ref{thm:4.1}, and denote its quotient family by \(\mathcal P\). Let its mean entropy be \(H>0\) and ratio be \(\delta\ge d\). Its sparsity is at most \(2\lambda\), while its convolution family \(\mathcal R\) is \(4\lambda\)-sparse. Form the posterior family \(\mathcal Q\) from this new quotient family. Lemma~\ref{thm:2.2} makes it \(2\lambda\)-sparse. Both mean entropies in \eqref{eq:31} are positive.

Set \[
 M=\max\left\{\frac{e(\mathcal R)}{h(\mathcal R)},
               \frac{e(\mathcal Q)}{h(\mathcal Q)}\right\},
 \qquad x=M-\delta.
\] Equations \eqref{eq:31}--\eqref{eq:32} give \(J\le2xH\), and hence \(x\ge0\). They also give
\begin{equation}\label{eq:40}
e(\mathcal Q)
 \ge\bigl[\delta^2-(2-\delta)x\bigr]H,
\end{equation}
because \(e(\mathcal R)\le M(2-\delta)H\) and \(J\ge0\). On the other hand, Lemma~\ref{thm:7.1} with \(\rho=2\lambda\), together with \(14\sqrt2<20\), yields \[
 e(\mathcal Q)\le72xH+20\sqrt\lambda\,\delta H.
\] Combining the bounds gives \[
 \delta^2-20\delta\sqrt\lambda\le(74-\delta)x.
\] The hypothesis \(\delta\ge d\ge40\sqrt\lambda\) makes the left side at least \(\delta^2/2\), so \(x\ge\delta^2/148\). Choose whichever of \(\mathcal R,\mathcal Q\) attains \(M\). It has sparsity at most \(4\lambda\), positive mean entropy, and ratio at least \(d+d^2/148\).
\end{proof}

\begin{theorem}[Weighted-family rigidity]\label{thm:8.2}

For \(0<\lambda<1\), every \(\lambda\)-sparse finite family with positive mean entropy satisfies
\begin{equation}\label{eq:41}
\frac{e(\mathcal P)}{h(\mathcal P)}
 \le\min\left\{1,\frac{320}{\log_2(1/\lambda)}\right\}.
\end{equation}

\end{theorem}

\begin{proof}

If the initial ratio is zero, there is nothing to prove. Otherwise write it as \(d\in(0,1]\), and let \(N=\lceil149/d\rceil\). Whenever a growth step is permitted, the successive ratios satisfy
\begin{equation}\label{eq:42}
\begin{aligned}
 \frac1{d_{j+1}}
 &\le\frac1{d_j+d_j^2/148}\\
 &=\frac1{d_j}-\frac1{148+d_j}
 \le\frac1{d_j}-\frac1{149},
 \end{aligned}
\end{equation}
using \(d_j\le1\).

Suppose that \[
 \lambda\le\frac{d^2}{1600}\,4^{-N}.
\] After \(j<N\) steps the sparsity parameter is at most \(4^j\lambda\le d^2/1600\), while the ratio is at least \(d\). Thus Lemma~\ref{thm:8.1} permits all \(N\) steps. Each family remains finite and has positive mean entropy. But \eqref{eq:42} would imply \[
 \frac1{d_N}\le\frac1d-\frac N{149}\le0,
\] a contradiction. Therefore
\begin{equation}\label{eq:43}
\lambda>\frac{d^2}{1600}\,4^{-\lceil149/d\rceil}.
\end{equation}
Taking base-two logarithms and multiplying by \(d\), \[
 \begin{aligned}
 d\log_2(1/\lambda)
 &<d\log_2 1600+2d\log_2(1/d)+2d\lceil149/d\rceil\\
 &\le298+(\log_2 1600+2)d+2d\log_2(1/d)\\
 &<312<320.
 \end{aligned}
\] Here \(d\le1\) and \(d\log_2(1/d)\le1/(e\ln2)\). This proves \eqref{eq:41}, while the other bound in the minimum is \eqref{eq:14}.
\end{proof}

Applying Theorem~\ref{thm:8.2} to a singleton family proves Theorem~\ref{thm:1.2}.

\section{Quotient consequences and a lower bound}\label{sec:9}

\subsection{The one-input theorem}\label{sec:9.1}

We record the elementary variational step for families, since it also gives the two-input consequence.

\begin{lemma}[Variational sparsification]\label{thm:9.1}

For any finite family \(\mathcal P\) and any \(\lambda>0\), there is a subspace \(G\) such that
\begin{equation}\label{eq:44}
\dim G\le\frac{h(\mathcal P)}\lambda,
 \qquad q_G\mathcal P\text{ is }\lambda\text{-sparse}.
\end{equation}

\end{lemma}

\begin{proof}

Minimize \(h(q_G\mathcal P)+\lambda\dim G\) over all subspaces. Comparison with zero gives the dimension bound. Comparing the minimizer with the inverse image of every further quotient subspace \(W\) gives \[
 F_{q_G\mathcal P}(W)\le\lambda\dim W.
\] This is the required sparsity.
\end{proof}

\begin{proof}[Proof of Theorem~\ref{thm:1.1}]

Set \[
 \theta=\frac\varepsilon{1+\varepsilon},\qquad
 \lambda=2^{-320/\theta},
\] and apply Lemma~\ref{thm:9.1} to the singleton family consisting of \(X\). The selected subspace satisfies \eqref{eq:2}, and \(Z=q_GX\) is \(\lambda\)-sparse. If \(H(Z)=0\), the entropy conclusion is immediate. Otherwise Theorem~\ref{thm:8.2} gives \[
 E(Z)\le\theta H(Z),
\] which is equivalent to \(H(Z)\le(1+\varepsilon)D(Z)\). Finally, \[
 \begin{aligned}
 D(X)&=I(X+X';X'),\\
 D(q_GX)&=I(q_G(X+X');q_GX')\le D(X)
 \end{aligned}
\] by data processing. This proves \eqref{eq:3}.
\end{proof}

\textbf{Comparison with the prior bound.} To make the distinction in Section~\ref{sec:1.1} precise, consider instead applying \cite[Theorem 1.2]{ref2} to two independent copies of a current law \(Y\). At cost at most \(2L(\alpha)H(Y)\), it produces a quotient \(Z\) with \(E(Z)\le2\alpha H(Y)\). Set \(\theta=\varepsilon/(1+\varepsilon)\) and \(\alpha=\theta/4\). If \(H(Y)\le(1+\varepsilon)D(Y)\), stop without taking a further quotient. Otherwise apply the quoted theorem. If \(H(Z)\ge H(Y)/2\), then \[
 E(Z)\le\frac\theta2H(Y)\le\theta H(Z),
\] so \(Z\) is intrinsically good and the procedure stops. If not, repeat with \(Z\), whose entropy has halved. A step from a law that is not already good must remove a nonzero subspace: the zero subspace would imply \(E(Y)\le\theta H(Y)/2\), contradicting \(E(Y)>\theta H(Y)\). Thus the process terminates in finite dimension. Lifting the successive kernels and summing their dimensions gives a total cost at most \[
 2L(\alpha)\bigl(H(X)+\tfrac12H(X)+\tfrac14H(X)+\cdots\bigr)
 \le4L(\alpha)H(X).
\] Data processing for \(D\) yields the original-input conclusion as well. Hence the earlier theorem already gives the qualitative conclusion of Theorem~\ref{thm:1.1}, with coefficient \[
 4L\!\left(\frac{\varepsilon}{4(1+\varepsilon)}\right)
 =\exp\bigl(O(\varepsilon^{-1}\log(1/\varepsilon))\bigr).
\] The present argument improves that scale to \(\exp(O(1/\varepsilon))\).

\subsection{Two-input almost-additivity}\label{sec:9.2}

\begin{corollary}\label{thm:9.2}

For \(0<\alpha<1\) and independent laws \(U,V\) on a common binary vector space, there is a subspace \(G\) with
\begin{equation}\label{eq:45}
\dim G\le\frac12\,2^{320/\alpha}\bigl(H(U)+H(V)\bigr)
\end{equation}
such that
\begin{equation}\label{eq:46}
H(q_GU+q_GV)
 \ge(1-\alpha)\bigl(H(q_GU)+H(q_GV)\bigr).
\end{equation}
In particular, the additive error is at most \(\alpha(H(U)+H(V))\). For the same subspace, \[
 H(q_GU)\le H(U+V)-H(V)
             +\alpha\bigl(H(q_GU)+H(q_GV)\bigr),
\] and the analogous inequality holds with \(U,V\) interchanged.

\end{corollary}

\begin{proof}

Let \(\mathcal P\) consist of the laws of \(U,V\), each with weight \(1/2\), and take \(\lambda=2^{-320/\alpha}\) in Lemma~\ref{thm:9.1}. This gives \eqref{eq:45}. Write \(\mathcal P_0=q_G\mathcal P\). If its mean entropy is zero, the conclusion is trivial. Otherwise Theorem~\ref{thm:8.2} gives \[
 e(\mathcal P_0)\le\alpha h(\mathcal P_0).
\] The two mixed ordered pairs have total weight \(1/2\), and the two self-interaction terms are nonnegative. Therefore \[
 \begin{aligned}
 s(q_GU,q_GV)&\le2e(\mathcal P_0)\\
 &\le\alpha\bigl(H(q_GU)+H(q_GV)\bigr),
 \end{aligned}
\] which is \eqref{eq:46}. Quotient entropy does not exceed input entropy, giving the original-input error bound. Finally, data processing gives \[
 \begin{aligned}
 H(q_GU+q_GV)-H(q_GV)
 &=I(q_GU;q_G(U+V))\\
 &\le I(U;U+V)=H(U+V)-H(V).
 \end{aligned}
\] Adding \(s(q_GU,q_GV)\) and using \eqref{eq:46} proves the displayed mixed-input bound. Interchanging the inputs gives its symmetric version.
\end{proof}

\subsection{Exponential dependence is necessary}\label{sec:9.3}

\begin{proposition}\label{thm:9.3}

Suppose \(C(\varepsilon)\) is valid for all binary-vector laws in the assertion \[
 \dim G\le C(\varepsilon)H(X),\qquad
 H(q_GX)\le(1+\varepsilon)D(X).
\] Then \eqref{eq:4} holds.

\end{proposition}

\begin{proof}

Fix \(c\in(0,2)\), and set \(p=2^{-c/\varepsilon}\). Take \(X\) to be a single Bernoulli bit with parameter \(p\). Write \[
 h_2(p)=-p\log_2p-(1-p)\log_2(1-p).
\] As \(p\downarrow0\), elementary expansion gives \[
 \begin{aligned}
 h_2(p)&=p\log_2(1/p)+p\log_2e+O(p^2),\\
 2h_2(p)-h_2(2p(1-p))&=2p+O\bigl(p^2\log(1/p)\bigr).
 \end{aligned}
\] Consequently \[
 \frac{E(X)}{H(X)}=\frac{2\varepsilon}{c}(1+o(1))
                    >\frac\varepsilon{1+\varepsilon}
\] for all sufficiently small \(\varepsilon\). The zero subspace does not satisfy the entropy condition, and the only other subspace has dimension one. Hence \(C(\varepsilon)\ge1/h_2(p)\), so \[
 \liminf_{\varepsilon\downarrow0}
 \varepsilon\log_2 C(\varepsilon)\ge c.
\] Let \(c\uparrow2\).
\end{proof}

Combining Proposition~\ref{thm:9.3} with Theorem~\ref{thm:1.1} shows that the correct unrestricted scale is exponential in \(1/\varepsilon\). It leaves a substantial gap between the lower leading constant two and the upper constant supplied by the proof.

\subsection{Concentration of moderate-doubling sets}\label{sec:9.4}

The following consequence makes the additive-combinatorial content explicit. It gives the conclusion of \cite[Theorem 1.1]{ref2} with an exponential, rather than \(\exp(O(\alpha^{-1}\log(1/\alpha)))\), dimension coefficient.

\begin{corollary}[Subspace-coset concentration]\label{thm:9.4}

Let \(0<\alpha<\eta\le1\), and let \(A\subseteq\mathbb F_2^m\) be nonempty with \(|A+A|\le|A|^{2-\eta}\). There is a subspace \(G\) such that \[
 \dim G\le2^{320/\alpha}\log_2|A|
\] and \[
 \frac1{|A|}\sum_{a\in A}\log_2|A\cap(a+G)|
 \ge\frac{\eta-\alpha}{1-\alpha}\log_2|A|.
\] In particular, some coset of \(G\) intersects \(A\) in at least \(|A|^{\eta-\alpha}\) points.

\end{corollary}

\begin{proof}

Let \(X\) be uniform on \(A\), so \(H(X)=\log_2|A|\). Since \(X+X'\) is supported on \(A+A\), \[
 D(X)\le\log_2|A+A|-\log_2|A|
       \le(1-\eta)\log_2|A|.
\] Apply Theorem~\ref{thm:1.1} with \(\varepsilon=\alpha/(1-\alpha)\). The stated dimension bound follows, and \[
 H(q_GX)\le\frac{1-\eta}{1-\alpha}\log_2|A|.
\] The conditional law of \(X\) on each occupied coset is uniform on its intersection with \(A\). Therefore \[
 \begin{aligned}
 \frac1{|A|}\sum_{a\in A}\log_2|A\cap(a+G)|
 &=H(X\mid q_GX)\\
 &=\log_2|A|-H(q_GX)\\
 &\ge\frac{\eta-\alpha}{1-\alpha}\log_2|A|.
 \end{aligned}
\] At least one term is no smaller than the average, and \((\eta-\alpha)/(1-\alpha)\ge\eta-\alpha\). This proves the last assertion.
\end{proof}

\section{Further questions}\label{sec:10}

The most immediate quantitative question is the optimal leading constant in the unrestricted exponential bound. The constants in the posterior estimate and ratio-growth argument were chosen for transparency rather than optimization. Sharpening them alone would not identify the extremal laws.

It is also natural to seek a constructive version with controlled running time. The present proof allows arbitrary finite family enlargement and global subspace minimization. Neither the number of posterior components nor the cost of these minimizations is bounded in a form that gives an efficient quotient-finding procedure.

The weighted-family formulation separates the information that survives conditioning from assertions that need not survive componentwise. Whether a comparable fixed-cost sparsity argument extends to other ambient groups, or to other entropy functionals, requires additional structure; no such extension is claimed here.

\section*{Acknowledgment}\label{acknowledgment}
\addcontentsline{toc}{section}{Acknowledgment}

OpenAI GPT-6 Astra and other AI-assisted tools were used extensively for proof development, checking, literature retrieval, drafting, and computational diagnostics during the preparation of this manuscript. Such checks do not constitute independent human peer review or formal proof-assistant verification. The proof of the main theorem is given in full, apart from the explicitly cited entropic Marton theorem; finite diagnostics are not used as premises of the proof. Responsibility for the submitted content rests with the named human author.


\begin{thebibliography}{9}

\bibitem{ref1}

W. T. Gowers, B. Green, F. Manners, and T. Tao, \emph{On a conjecture of Marton}, arXiv:2311.05762v1, 9 November 2023. Theorem and lemma numbering in this paper refers to version 1. \url{https://arxiv.org/abs/2311.05762v1}.

\bibitem{ref2}

A. Cohen and D. Zakharov, \emph{Moderate-doubling sets in \(\mathbb F_2^n\) intersect subspaces}, arXiv:2512.19884v1, submitted 22 December 2025. \url{https://arxiv.org/abs/2512.19884v1}.

\bibitem{ref3}

E. Abbe, C. Sandon, V. Shashkov, and M. Viazovska, \emph{Polynomial Freiman--Ruzsa, Reed--Muller codes and Shannon capacity}, arXiv:2411.13493v2, revised 25 February 2026. \url{https://arxiv.org/abs/2411.13493v2}.

\bibitem{ref4}

J. Park, \emph{Entropy quotients of binary products and structured linear images}, unpublished manuscript dated 28 October 2026, Section 12. The relevant variational reduction is reproduced and proved in Section~\ref{sec:9.1} of the present article.

\end{thebibliography}
\end{document}